\documentclass[submission]{eptcs}

\usepackage{breakurl}

\sffamily

\usepackage{color}
\usepackage{a4wide}
\usepackage{url}
\usepackage{xypic}
\usepackage{enumitem}
\usepackage{tikz}

\def\emph{\textbf}

\usepackage{amsmath,amscd,amsthm,amssymb}
\def\sfrac#1#2{#1/#2}

\def\Vorobev{Vorob'ev}

\let\coprod=\undefined
\DeclareSymbolFont{cmsymbols}{OMS}{cmsy}{m}{n}
\DeclareSymbolFont{cmlargesymbols}{OMX}{cmex}{m}{n}
\DeclareMathSymbol{\coprod}{\mathop}{cmlargesymbols}{"60}

\theoremstyle{plain}
\newtheorem{theorem}{Theorem}[section]
\newtheorem{proposition}[theorem]{Proposition}
\newtheorem{corollary}[theorem]{Corollary}

\theoremstyle{definition}
\newtheorem{definition}[theorem]{Definition}
\newtheorem{example}[theorem]{Example}

\title{Extendability in the Sheaf-theoretic Approach: \\ Construction of Bell Models from Kochen-Specker Models}

\author{
Shane Mansfield
\qquad Rui Soares Barbosa
\institute{Quantum Group \\ Department of Computer Science\\
University of Oxford
}
\email{\{shane.mansfield,rui.soares.barbosa\}@cs.ox.ac.uk}
}

\date{\today}

\def\M{\mathcal{M}}

\def\E{\mathcal{E}}
\def\Ecd{\mathcal{E}^\triangledown}
\def\Bell{\mathsf{Bell}}
\def\down{\downarrow\!}
\def\supp{\mathsf{supp}}
\def\bigmeet{\bigwedge}

\def\KS{\mathsf{KS}}

\def\setdef#1#2{\left\{#1\mid #2\right\}}
\def\Forall#1{\forall_{#1}\boldsymbol{.}\;}
\def\Exists#1{\exists_{#1}\boldsymbol{.}\;}
\def\tuple#1{\langle #1 \rangle}

\def\supp{\mathsf{supp}}

\newenvironment{calculation}{\begin{eqnarray*}&&}{\end{eqnarray*}}
\def\just#1#2{\\ &#1& \rule{2em}{0pt} \{ \mbox{\rule[-.7em]{0pt}{1.8em} \footnotesize #2 \/} \} \nonumber\\ && }
\def\ejust#1{\\ &#1& \nonumber\\ && }

\begin{document}

\maketitle

\begin{abstract}
Extendability of an empirical model was shown by Abramsky \& Brandenburger to correspond in a unified manner to both locality and non-contextuality. We develop their approach by presenting a refinement of the notion of extendability that can also be useful in characterising the properties of sub-models. The refinement is found to have another useful application: it is shown that a particular, canonical extension, when well-defined, may be used for the construction of Bell-type models from models of more general kinds in such a way that the constructed model is equivalent to the original in terms of non-locality/contextuality. This is important since on practical and foundational levels the notion of locality in Bell-type models can more easily be motivated than the corresponding general notion of contextuality. We consider examples of Bell-type models generated from some standard examples of contextual models, including an entire class of Kochen-Specker-like models. This exposes an intriguing relationship between the simplest possible contextual model (the contextual triangle) and Popescu-Rohrlich no-signalling correlations.
\end{abstract}

\section{Introduction}

Any physical theory must make predictions for empirical observations. We will refer to any (possibly hypothetical) set of empirical observations, or any set of theoretical predictions for empirical observations, as an empirical model. In \cite{abramsky:11}, extendability of empirical models was seen to correspond in a unified manner to both locality and non-contextuality, an insight that has initiated diverse lines of research (e.g. \cite{abramsky:12c,abramsky:13, abramsky:12, abramsky:11a,barbosa:13}). We will develop this approach by introducing a refinement of the notion of extendability that captures the idea of partial approximations to locality/non-contextuality. This can be useful in characterising the properties of sub-models.

We will also consider another application of this refinement. Certain empirical models have measurements that can be partitioned into sites, and can be considered to abstract spatially distributed systems: these are the Bell-type models. We are especially interested in a particular, canonical extension, which, when well-defined, may be used for the construction of equivalent Bell-type models from models of the more general kind. On both foundational and practical levels, an advantage of having an equivalent Bell form of a contextual model is that it is much easier to motivate a notion of locality in a Bell scenario than the corresponding notion of non-contextuality in a more general type of measurement scenario, making non-local behaviour all the more striking (a more detailed discussion is contained in \cite{mansfield:13t}). This was realised by Bell, who had observed a similar result \cite{bell:66}\footnote{Written before but eventually published later than \cite{bell:64}.} to that of Kochen \& Specker \cite{kochen:75} before going on to prove his more well-known non-locality theorem \cite{bell:64} (this is also discussed in \cite{mermin:93}). A further advantage is that non-locality can be exploited as an information theoretic resource \cite{barrett:05a}, whereas contextuality has yet to be developed for such purposes.

We find equivalent Bell-type models for many familiar examples of contextual models; in particular, the entire family of Kochen-Specker-like models, which includes, for example, the contextual triangle of Specker's parable \cite{liang:11} and the $18$-vector Kochen-Specker model \cite{cabello:96}. One connection that arises is that the equivalent Bell-type model for the contextual triangle is essentially a folding of several Popescu-Rohrlich boxes \cite{popescu:94}. The Peres-Mermin square \cite{mermin:93} is also treated. This represents a step in the direction of proposing equivalent Bell tests for contexuality results, though an important issue that remains to be addressed is that of quantum realisability.

\subsection*{Outline}
We work within the sheaf-theoretic framework of Abramsky \& Brandenburger \cite{abramsky:11}, which provides an elegant language for dealing with empirical models. We begin by providing a brief summary of the relevant ideas in section \ref{framework}, with particular attention to the kinds of models and measurement scenarios with which we will be concerned here. In section \ref{ssec:extension}, the refined notion of extendability, which in a certain limit recovers the usual notion, is introduced. The construction of Bell-type models is explained in section \ref{ssec:construction}; in section \ref{ssec:ksconstruction} it is shown that this construction can be carried out for the entire family of Kochen-Specker-like models; and examples are given in section \ref{examples}. Finally, we discuss open questions and ongoing lines of research.

\section{The Sheaf-theoretic Framework}\label{framework}

We assume sets $X$ of measurements and $O$ of outcomes. There is an additional structure on the set of measurements, a cover $\M$ over $X$, which specifies the sets of compatible measurements: we think of these as sets of measurements that can be performed jointly. In quantum mechanics, for example, this structure would arise as the commutative subalgebras of the algebra of observables.

\begin{definition}\label{def:mscen}\index{measurement scenario}
We will refer to $(X,O,\M)$ as a \emph{measurement scenario}.
\end{definition}

We will restrict our attention to finite measurement scenarios. Sets in the down-closure $\down\M$ will be referred to as \emph{contexts}\index{context} and will be denoted by the letters $U,V, \dots$; elements of the cover $\M$ itself will usually be referred to as \emph{maximal contexts}\index{maximal context|see{context}}\index{context!maximal context} and will be denoted by the letters $C,D, \dots$.

The \emph{event sheaf}\index{event sheaf} $\E: \mathcal{P}^{\mathrm{op}}(X) \rightarrow \mathbf{Set}$ is defined by $\E(U):=O^U$ for each $U \subseteq X$; i.e.~$\E(U)$ contains all functional assignments of outcomes to the measurements in $U$. In order to describe an empirical model we must specify a probability distribution over the assignments $\E(C)$ for each maximal context $C \in \M$. This can be achieved by composing $\E$ with the distribution functor $\mathcal{D}_R: \mathbf{Set} \rightarrow \mathbf{Set}$, which takes a set to the set of $R$-distributions over it, where $R$ is some semiring. Probability distributions are obtained when $R = \mathbb{R}^+$\index{negative probability}, the non-negative reals. More generally, it can be useful to consider other kinds of distributions: for example `negative probability' ($R = \mathbb{R}$) or `possibilistic' ($R=\mathbb{B}$, the Boolean semiring) distributions. The composition of the two functors, $\mathcal{D}_R \E$, is a presheaf in which restriction is given by marginalisation of distributions. Now, an empirical model can be specified by a family of distributions $\{ e_C \}_{C \in \M}$, where each $e_C \in \mathcal{D}_R \E (C)$.

To avoid confusion between sections of the event sheaf $\E$ and the presheaf $\mathcal{D}_R \E$, we will refer to sections of the former as \emph{assignments} throughout, since they are understood to assign outcomes to measurements.

We build the property of no-signalling into our models by imposing the condition that the marginals of the distributions $\{e_C\}_{C \in \M}$ specifying an empirical model agree wherever contexts overlap; i.e.
\[
\Forall{C,D \in \M}  e_C|_{C \cap D} = e_D|_{C \cap D}.
\]
This implies that there are well-defined distributions $e_U$ for all $U \in \down\M$, since we obtain the same distribution no matter which maximal context we marginalise from. This is \emph{compatibility} in the sense of the sheaf condition.
This property is true of all models arising from quantum mechanics \cite{abramsky:11},
and it corresponds to the usual no-signalling\index{no-signalling} \cite{ghirardi:80}
when the measurement scenario is of the Bell type (section \ref{ssec:types}).

\begin{definition}\index{empirical model!sheaf-theoretically}
An \emph{empirical model} $e$ over a measurement scenario $(X,O, \M)$ is a compatible family of $R$-distributions
\[
\{ e_C \}_{C \in \M},
\]
with $e_C \in \mathcal{D}_R \E (C)$ for each $C \in \M$.
\end{definition}

We can use tables as a convenient way of representing empirical models. The following example illustrates how such a table is anatomised in the sheaf-theoretic language.

\begin{example}[The Bell-CHSH Model \cite{clauser:69,bell:87}]\label{exa:bell}
In this bipartite empirical model, Alice can choose between two measurements $A$ and $A'$ and Bob can choose between $B$ and $B'$, each of which can take outcomes $0$ or $1$, with probabilities as represented in the following table. 
{\rm
\begin{center}
\begin{tabular}{p{3pt}l|cccc} 
 &~&  $00$ & $01$ & $10$ & $11$ \\ \hline
$A$&$B$ &  \sfrac{1}{2} & $0$ &  $0$ & \sfrac{1}{2} \\
$A$&$B'$ &   \sfrac{3}{8} & \sfrac{1}{8} &  \sfrac{1}{8} &  \sfrac{3}{8} \\
$A'$&$B$ &  \sfrac{3}{8} & \sfrac{1}{8} &  \sfrac{1}{8} &  \sfrac{3}{8} \\
$A'$&$B'$ &   \sfrac{1}{8} &  \sfrac{3}{8} &  \sfrac{3}{8} & \sfrac{1}{8} 
\end{tabular}
\end{center}
}
The measurement scenario is described by $X = \{ A,A',B,B' \}, O = \{0,1\}$ and
\[
\M = \{ \{A,B\}, \quad \{A,B'\}, \quad \{A',B\}, \quad \{A',B'\} \}.
\]
The labels for the rows correspond to the maximal contexts, and the cells of each row $C \in \M$ (ignoring the entries for now) correspond to the assignments $\E(C)$. For example,
\begin{equation*}
\E(\{A,B\}) = \{ AB\mapsto 00, \quad AB \mapsto 01, \quad
 AB \mapsto 10, \quad AB \mapsto 11 \}.
\end{equation*}
The entries of each row specify the probability distrubution over these assignments (the joint outcomes). For example, the first row of the table
{\rm
\begin{center}
\begin{tabular}{p{3pt}l|cccc} 
&~ &  $00$ & $01$ & $10$ & $11$ \\ \hline
$A$&$B$ &  \sfrac{1}{2} &  $0$ &  $0$ &  \sfrac{1}{2} 
\end{tabular}
\end{center}}
corresponds to the distribution
\[
e_{\{A,B\}} \in \mathcal{D}_{\mathbb{R^+}} \E (\{A,B\}) .
\]
The sheaf-theoretic empirical model $e$ obtained in this way is well-defined since it arises from quantum mechanics and is therefore necessarily compatible (no-signalling).
\end{example}

Many familiar models are considered in the sheaf-theoretic setting in \cite{abramsky:11} and \cite{abramsky:11a}, including models used for the non-locality/contextuality results of Hardy \cite{hardy:93}, Kochen \& Specker \cite{kochen:75}, Greenberger, Horne \& Zeilinger \cite{greenberger:90} and Peres \& Mermin \cite{peres:91}. 

\subsection{Extendability}\label{ssec:extendability}

An important feature of the framework is that it is general enough to provide a unified approach to non-locality\index{non-locality} and contextuality. The main result of \cite{abramsky:11} is the following theorem.

\begin{theorem}[Abramsky \& Brandenburger]\label{thm:ab}
An empirical model can be realised by a factorisable hidden variable model if and only if the model is extendable to a global section.
\end{theorem}

By \emph{factorisability} it is meant that, when conditioned on any particular value of the hidden variable, the probability assigned to a joint outcome should factor as the product of the probabilities assigned to individual outcomes. For Bell scenarios (section \ref{ssec:types}) this corresponds exactly to Bell locality \cite{bell:64}. On the other hand, a model is said to be \emph{extendable to a global section} precisely when there exists a distribution $d \in \mathcal{D}_R \E(X)$ over global assignments such that $d|_C = e_C$ for all $C \in \M$. This corresponds to non-contextuality in the sense of the Kochen-Specker Theorem \cite{kochen:75}. In fact, the global assignments provide a canonical hidden variable space \cite{brandenburger:11,brandenburger:13}.

In the sheaf-theoretic language, then, locality and non-contextuality are characterised in a unified manner by the existence of global sections. Contextuality will therefore sometimes be used as a general term which is assumed to include non-locality. This insight has already led to many interesting results \cite{abramsky:12c,abramsky:11,abramsky:13, abramsky:12, abramsky:11a,barbosa:13}. Non-locality and contextuality are characterised by obstructions to the existence of global sections.

This greatly generalises earlier work by Fine \cite{fine:82}, which showed in certain bipartite Bell-type measurement scenarios\footnote{(2,2,2) Bell scenarios; c.f. section \ref{ssec:types}.} that for any local hidden variable model there exists an equivalent deterministic local hidden variable model.

\subsection{Possibilistic Collapse}\label{ssec:posscollapse}

It has already been mentioned that it can be of interest to consider `possibilistic' models (Boolean distributions) in which the relevant information consists in knowing whether various outcomes are possible or not, rather than their probabilities. Logical proofs of non-locality \cite{abramsky:10,mansfield:11} (i.e. `proofs without inequalities' of the kind found, for example, in the GHZ \cite{greenberger:90} and Hardy \cite{hardy:93} arguments) rely on such models, and many of the models used for contextuality arguments are also possibilistic in nature (e.g. \cite{cabello:96,peres:91}).

In a possibilistic empirical model the distributions are Boolean\index{Boolean semiring}; i.e.~the semiring is $R = \mathbb{B} = \left( \{0,1\}, \vee ,0, \wedge,1 \right)$ where the truth value `$1$' is understood to denote `possible' and `$0$' to denote `impossible', and `$\vee$' and `$\wedge$' are the logical `and' and `or' operations. We can obtain a possibilistic model from any empirical model via the process of \emph{possibilistic collapse}, which turns any empirical model into a possibilistic one by conflating all non-zero probabilities to the Boolean `$1$'. More carefully, its action is described by the natural transformation $\gamma: \mathcal{D}_{\mathbb{R}^+} \rightarrow  \mathcal{D}_{\mathbb{B}}$ induced by the function
\[
h: \mathbb{R}^+ \rightarrow \mathbb{B},
\quad p \mapsto
\begin{cases}  
  0 & \text{if $p=0$} \\
  1 & \text{if $p\neq0$} \\
\end{cases}.
\]

\begin{example}
The Bell-CHSH model of Example \ref{exa:bell} collapses to the following possibilistic model.
\begin{center}
\begin{tabular}{p{3pt}l|cccc} 
&~ &  $00$ & $01$ & $10$ & $11$ \\ \hline
$A'$&$B$ &  $1$ &  $0$ &  $0$ & $1$ \\
$A$&$B'$ &  $1$ & $1$ & $1$ & $1$ \\
$A$& $B'$ & $1$ & $1$ & $1$ & $1$ \\
$A'$ & $B'$ & $1$ & $1$ & $1$ &$1$ 
\end{tabular}
\end{center}
\end{example}

We introduce a notation that will be extremely useful in dealing with possibilistic models. The \emph{support}\index{support!of a distribution} of a distribution $d$ over $Y$ is the set
\[ \supp(d) := \setdef{y \in Y}{d(y) \neq 0}.\]
We generalise the support notation, by defining, for any $U \subseteq X$,
\[ S_e(U) := \setdef{s \in \E(U)}{\Forall{C\in\M} s|_{C\cap U} \in \supp(e_C|_{C \cap U})}.\]
That is to say, the set $S_e(U)$ contains all functional assignments of outcomes to the measurements $U$ that are consistent with the model $e$. In particular, the set $S_e(X)$ contains all the global assignments\index{assignment!consistency} that are consistent with the model $e$. It is a generalisation in the sense that, for each $U \in \down\M$, we have $S_e(U) = \supp(e_U)$. It can easily be shown that $S_e:\mathcal{P}(X)^{\mathrm{op}} \rightarrow \mathbf{Set}$ defines a sub-presheaf of the sheaf of events.

The possibilistic content of an empirical model is that which is available at the level of the support of the distributions of which it is made up. That is because a Boolean distribution can be equivalently represented by its support: i.e.~there is a bijection
\[
\supp(d) \cong \setdef{y \in Y}{\gamma d(y) = 1}
\]
between the distributions $\mathcal{D}_\mathbb{B}(Y)$ and the non-empty subsets of $Y$,
and therefore
\[\{ S_e(C)\}_{C\in\M} \cong \{ \gamma e_C \}_{C\in\M}.\]

\subsection{A Hierarchy of Contextuality}\label{ssec:posscontextuality}

\subsubsection*{Logical Contextuality}

At the possibilistic level, for any empirical model $e$, we can pose the problem of whether $\gamma e$ is extendable to a global section. Then the problem is to find a Boolean distribution over the global assignments $\E(X)$ which restricts to $\gamma e_C$ for each maximal context $C$. If such a distribution exists we will say that $e$ is \emph{possibilistically extendable} (to a global section).

Using the notation introduced in section \ref{ssec:posscollapse}, we are interested in the existence of a Boolean distribution $d \in D_\mathbb{B}\E(X)$ for which the following conditions hold.
\begin{enumerate}
  \item\label{cond1} $\supp(d) \subseteq S_e(X)$; i.e.~all global assignments in $\supp(d)$ are consistent with the empirical model.
  \item\label{cond2} $\Forall{C\in\M}\Forall{t \in S_e(C)}\Exists{s \in \supp(d)} t = s|_C$; i.e.~any possible local assignment can be obtained as the restriction of some global assignment in $\supp(d)$.
\end{enumerate}
In short,
\[
S_e(C) = \setdef{s|_C}{s \in \supp(d)}
\]
for each $C\in\M$.

\begin{definition}\index{non-locality!logical non-locality}\index{contextuality!logical contextuality}
If an empirical model is not possibilistically extendable to a global section we say that the model is \emph{logically contextual} (or \emph{logically non-local} when appropriate).
\end{definition}

These are the empirical models that admit `logical' proofs of non-locality. Some models can be non-local or contextual without exhibiting the properties at the possibilistic level: an example is the Bell-CHSH model. However, it can be shown that a probabilistic model that exhibits logical contextuality at the possibilistic level is necessarily contextual at the probabilistic level, too \cite{abramsky:11}. Logical contextuality is therefore a strictly stronger form of contextuality. Many familiar empirical models exhibit logical non-locality or contextuality, including the Hardy model \cite{hardy:92,hardy:93}. A recent result \cite{ying:13} even indicates that for any multipartite qubit state there exists some choice of measurements that will give rise to logical non-locality.

\subsubsection*{Strong Contextuality}

Recall that $S_e(X)$ consists of those global assignments that are \emph{consistent} with the support of $e$; i.e.~whose restrictions to every context of compatible observables are possible according to $e$.
These are the only global assignments that could be taken to be possible. It has already been observed that if a possibilistic extension $d \in D_\mathbb{B}\E(X)$ exists then $\supp(d) \subseteq S_e(X)$, and it is clear that in this case $S_e(X)$ is also a possibilistic extension of $e$. This follows from condition \ref{cond2}: if any possible local assignment arises as the restriction of an assignment in $\supp(d)$ then, since $\supp(d) \subseteq S_e(X)$, it arises as a restriction of an assignment in $S_e(X)$.
For this reason, $S_e(X)$ can be regarded as providing a canonical candidate\index{extendability!canonical extension} for a possibilistic extension of the empirical model $e$.

In general, the set $S_e(X)$ of consistent global assignments can fail to determine an extension of the empirical model $e$
if it isn't large enough to account for all `local' assignments\index{assignment!`local'} that are possible in $e$;
that is, if there exists some assignment $s \in S_e(C)$ on some maximal context $C \in \M$ which does not arise as a restriction of a global assignment in $S_e(X)$, as in the example of the Hardy model.
However, the extreme case happens when $S_e(X)$ is empty (then, $S_e(X)$ does not even determine a distribution over $\E(X)$).
This means that there is no global assignment that is consistent with the support of $e$.

\begin{definition}
If no global assignment is consistent with the support of the model $e$, i.e. $S_e(X) = \emptyset$, we say that $e$ is \emph{strongly contextual} (or \emph{strongly non-local}, when appropriate).
\end{definition}

Note that to have non-empty $S_e(X)$ is a weaker property than
possibilistic extendability:
it is simply asking for the existence of some global assignment consistent with the support of $e$.
Correspondingly, the negative property is stronger than possibilistic non-extendability (logical contextuality or non-locality). Some of these ideas will be generalised in section \ref{ssec:extension}.

The Hardy model is logically non-local but not strongly non-local. Strong contextuality is displayed by many models, however, including the GHZ-Mermin model \cite{mermin:90a,mermin:90}, the 18-vector Kochen-Specker model, the Peres-Mermin `magic square' \cite{mermin:93,peres:91} and the Popescu-Rohrlich correlations \cite{popescu:94}.

We thus arrive at a strict hierarchy of contextuality: strong contextuality is stricly stronger than logical contextuality, which in turn is strictly stronger than (probabilistic) contextuality.

\subsection{Bell Scenarios}\label{ssec:types}

Bell scenarios are measurement scenarios that can be thought of as abstracting spatially distributed systems (see figure \ref{fig:bell}). More precisely, $(n,k,l)$ Bell scenarios can be thought of as $n$-partite models, in which each party may choose to perform one of $k$ different measurements, each of which could have $l$ possible outcomes. For example, the model arising from the Bell-CHSH theorem in example \ref{exa:bell} or the Hardy model are both $(2,2,2)$ models. An example of a $(3,2,2)$ model is that which arises from the GHZ-Mermin non-locality argument. Recall that in such scenarios extendability corresponds to the usual notion of Bell locality. Models over measurement scenarios of this kind will be referred to as \emph{Bell-type models} (or often simply \emph{Bell models}).

These are measurement scenarios $(X,O,\M)$ for which the set of measurements can be written as a disjoint union $X = \coprod_{i=1}^n X_i$ such that the maximal contexts are given by the cartesian product $\M = \prod_{i=1}^n \, X_i$.
We define $l:=\left| O \right|$ and $k:= \max_{1\leq i\leq n} \left| X_i\right|$.

As a technical remark, there is a slight abuse of notation here.
Elements of $X = \coprod_{i=1}^n X_i$
are of the form $\tuple{x,i}$ where $i \in \{1,\ldots,n\}$ identifies the site and $x \in X_i$.
An element of the cover $\prod_{i=1}^n \, X_i$, which is an $n$-tuple $\tuple{x_1,\ldots,x_n}$ with each $x_i \in X_i$,
can be seen as a subset of $\coprod_{i=1}^n X_i$ if we interpret it as
$\{\tuple{x_1,1},\ldots,\tuple{x_n,n}\}$.
We will denote it as a tuple, however, to simplify notation.
Therefore, the maximal contexts of a Bell scenario are the sets of measurements that contain one measurement from each site.

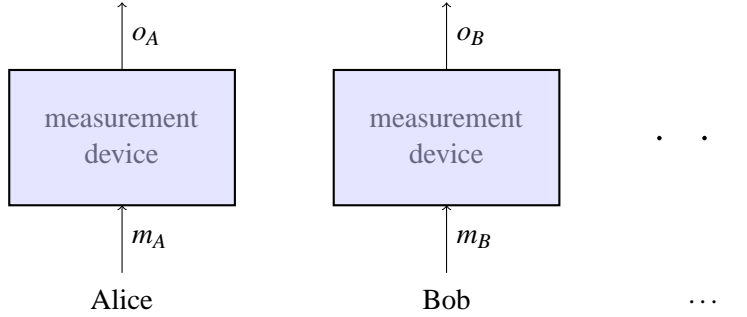
\begin{figure}
\caption{\label{fig:bell} A Bell scenario.}
\begin{center}
\begin{tabular}{ccc}
\begin{tikzpicture}[scale=6]
  \path (0,0) coordinate (A1);
  \path (0,0.3) coordinate (A2);
  \path (0.5,0.3) coordinate (A3);
  \path (0.5,0) coordinate (A4);
  \path (0.25,-0.15) coordinate (B1);
  \path (0.25,0) coordinate (B2);
  \path (0.25,0.3) coordinate (B3);
  \path (0.25,0.45) coordinate (B4);
  \path (0.25,0.15) node {\begin{tabular}{c} measurement \\ device \end{tabular}};
  \begin{scope}[ thick,black]
    \filldraw[fill=blue!20,fill opacity=0.5] (A1) -- (A2) -- (A3) -- (A4) -- cycle;
  \end{scope} 
  \draw[->] (B1) to
    node[midway,right] {$m_A$} (B2);  
  \draw[->] (B3) to
    node[midway,right] {$o_A$} (B4);
\end{tikzpicture}
\qquad & \qquad
\begin{tikzpicture}[scale=6]
  \path (0,0) coordinate (A1);
  \path (0,0.3) coordinate (A2);
  \path (0.5,0.3) coordinate (A3);
  \path (0.5,0) coordinate (A4);
  \path (0.25,-0.15) coordinate (B1);
  \path (0.25,0) coordinate (B2);
  \path (0.25,0.3) coordinate (B3);
  \path (0.25,0.45) coordinate (B4);
    \path (0.25,0.15) node {\begin{tabular}{c} measurement \\ device \end{tabular}};
  \begin{scope}[ thick,black]
    \filldraw[fill=blue!20,fill opacity=0.5] (A1) -- (A2) -- (A3) -- (A4) -- cycle;
  \end{scope} 
  \draw[->] (B1) to
    node[midway,right] {$m_B$} (B2);  
  \draw[->] (B3) to
    node[midway,right] {$o_B$} (B4);
\end{tikzpicture}
\qquad & \qquad
\begin{tikzpicture}[scale=6]
 \path (0.1,0) coordinate (A1);
  \fill[black]   (.1,.3) circle(.005) ;  \fill[black]   (.2,.3) circle(.005) ;  \fill[black]   (.3,.3) circle(.005) ;
\end{tikzpicture}
\\ Alice \qquad & \qquad Bob \qquad & \qquad $\cdots$
\end{tabular}
\end{center}
\end{figure}

\subsection{Kochen-Specker Models}

Of course, not all models are over measurement scenarios of the Bell type: for example, the state-independent model for the $18$-vector proof of the Kochen-Specker theorem\index{Kochen-Specker model} and the Peres-Mermin square. Both of these models make use of measurement scenarios that are of a more general form and cannot be partitioned into sites. Many `non-Bell' models fall into a general class of \emph{Kochen-Specker models}.

These are the possibilistic models on any measurement scenario for which, at each maximal context, an assignment is possible if and only if it maps a single measurement to the outcome $1$. The cohomological characterisation of strong contextuality \cite{abramsky:11a} is complete for certain classes of Kochen-Specker models. Models of this kind arise in a natural way from projective measurements (e.g. \cite{cabello:96}). In this case all maximal contexts contain the same number of measurements, which is fixed by the dimension of the Hilbert space. While Kochen-Specker models may be defined for any measurement cover, we will only consider covers with the property that maximal contexts contain the same number of measurements in this work.

\begin{definition}\label{def:ksfam}
Let $O=\{0,1\}$, and $\M$ be a cover of $X$ for which there exists $n \in \mathbb{N}$ such that $\Forall{C \in \M} \left| C\right| = n$.
The \emph{Kochen-Specker model} on the measurement scenario $(X,O,\M)$ is the possibilistic model $\{e_C\}_{C \in \M}$ defined by
\[
S_{e}(C)=S_{\KS}(C) := \setdef{s \in \E(C)}{o(s)=1},
\]
for all $C \in \M$, where $o(s) := |\setdef{x \in C}{s(x) = 1}|$ for any assignment $s \in \E(C)$.
\end{definition}

\begin{example}
The following is the simplest non-trivial example of a Kochen-Specker model. 
\begin{center}
\begin{tabular}{p{3pt}l|cccc}
&~& $00$ & $01$ & $10$ & $11$   \\ \hline
$A$&$B$ & $0$ & $1$ & $1$ & $0$ \\
$B$&$C$ & $0$ & $1$ & $1$ & $0$ \\
$C$&$A$ & $0$ & $1$ & $1$ & $0$ \\
\end{tabular}
\end{center}
We refer to this model as the \emph{contextual triangle}. It is the model that arises from Specker's parable \cite{seevinck:11,liang:11}. It has also appeared in a somewhat different context in \cite{abramsky:11}. We will return to this model in section \ref{examples}.
\end{example}

\section{No-signalling Extensions of Models}\label{ssec:extension}

We consider the problem
of extending an empirical model to a cover that allows increased compatibility of measurements. For notational convenience, in this section, we fix sets $X$ of measurements and $O$ of outcomes,
so that a measurement scenario can be identified by its cover $\M$ of maximal contexts alone. Also, when we refer to models in this section it will be assumed that we refer to possibilistic models, as introduced in section \ref{ssec:posscollapse}.

\begin{definition}
  Let $\M$ and $\M'$ be two measurement covers. We write $\M \preceq \M'$ when
  $\down \M \subseteq \down \M'$; i.e.
  \[
  \Forall{D \in \M} \quad \Exists{C \in \M'} \quad D \subseteq C.
  \]
\end{definition}

\begin{definition}\index{extendability!partial extendability}
Let $\M \preceq \M'$, and let $e$ be a model defined on $\M$.
A model $f$ on $\M'$ is said to \emph{extend} $e$ if \[\Forall{D \in \M} \qquad f_D = e_D.\]
When such an $f$ exists, we say that $e$ \emph{is extendable to $\M'$}.
\end{definition}

Note that the cover $\M_\top := \{X\}$, in which any subset of measurements is jointly compatible,
is larger than all other covers; i.e.~it is the top element in the poset of measurement covers.
Asking for extendability to the top cover amounts to asking for extendability in the usual sense\index{extendability}:
in other words, locality or non-contextuality.
The notion of extendability to any cover $\M' \prec \M_\top$ therefore captures partial approximations to the usual notion.
One cover that will be of particular interest later is the following.
\begin{definition}
For any cover $\M$ we define $n(\M) := \max_{C\in \M} \left| C \right|$ to be the maximum size of contexts in $\M$ (where no confusion arises as to the cover in question this will be simply denoted $n$). Then we can define another cover
\[
\mathcal{P}_{n}X := \setdef{Y \subseteq X}{|Y| = n}
\]
over $X$. It necessarily holds that $\M \preceq \mathcal{P}_{n}X$.
\end{definition}

We now consider a construction that provides a canonical candidate for an extension of a model to any larger cover (much like $S_e(X)$ for the usual notion of extendability)\index{extendability!canonical extension}. We note, however, that this will not necessarily yield a well-defined model. The idea is to allow every assignment except
those that are directly forbidden by the compatibility/no-signalling condition; i.e.~to allow every assignment in $e'$ that is \emph{consistent}
with the possible assignments in $e$.

\begin{definition}\index{extendability!canonical extension}
  Let $\M \preceq \M'$ and $e$ be a model on $\M$.
  For each $C \in \M'$ and $s \in \E(C)$, we define:
  \[ e'_C(s) := \bigmeet_{W \subseteq C, W \in \down \M} e_W (s|_W) .\]
  If $\{e'_C\}_{C \in \M'}$ is a well-defined model extending $e$,
  we call it the \emph{canonical extension} of $e$ to $\M'$,
  and say that $e$ is \emph{canonically extendable to} $\M'$.
\end{definition}

According to the definition, the support of the extended model $e'$ at each maximal context $C \in \M'$ is
\[ \supp(e'_C) =  \setdef{s \in \E(C)}{\Forall{W \in \down \M, W \subseteq C} \, s|_W \in S_{e}(W)} = S_e(C);\]
i.e.~the support contains those assignments on $C$ that are consistent with the model $e$. We can equivalently express this in a way that mentions only maximal contexts:
\[\supp(e'_C) = \setdef{s \in \E(C)}{ \Forall{D \in \M} \, \Exists{t \in \E(D)} t \in S_{e}(D) \land t |_{C \cap D} = s |_{C \cap D}}.\]
Clearly, for $\M' = \M_\top$,
\[
\supp(e'_X) = S_e(X);
\]
i.e.~the assignments deemed possible by $e'$ are precisely the global assignments consistent with $e$.

It has been shown that $S_e(X)$ provides a canonical local hidden variable space for the model $e$ \cite{brandenburger:11,brandenburger:13}. Although the present construction does not satisfy properties that are quite as strong, the next two propositions show
why this construction, when it yields a well-defined extension, can be regarded as canonical in some sense, especially
with regard to strong contextuality.
\begin{proposition}\label{prop:smaller}
  Let $\M \preceq \M'$, $e$ be a model on $\M$, and $f$ be a model on $\M'$ that extends $e$.
  Then, for all $C \in \M'$, \[ \supp(f_C) \subseteq \supp(e'_C)\] (i.e.~$f_C(s)=1$ implies $e'_C(s)=1$ for any $s \in \E(C)$).
\end{proposition}
\begin{proof}
Consider any maximal context $C \in \M'$ and any assignment $s \in \supp(f_C)$. Then, by virtue of $f$ being an extension of
$e$, all restrictions of $s$ to contexts in $\M$ must be consistent with $e$.
That is, for all $W \subseteq C$ with $W \in \down \M$, we have $s|_W \in S_e(W)$.
Then, by the definition of $e'$, it must be that $s \in \supp(e'_C)$.
\end{proof}
This tells us that any extension has less possible assignments than the canonical one. This is not surprising, since the canonical
construction picks out all the assignments that are consistent with the model $e$.
It is clear that, in the extreme,
$e'_C$ might fail to be a distribution for some $C \in \M'$; i.e.
\[ \supp(e'_C)=S_{e}(C) = \emptyset.\]
Obviously, in that case there can be no extension of $e$ to $\M'$ whatsoever.
We say that
$e$ is \emph{strongly non-extendable to $\M'$}, in analogy with the notion of strong contextuality.

The following proposition will be relevant for the construction of Bell models in section \ref{ssec:construction}.
It can be read as saying that $e'$ is the most conservative extension that can be made in terms of not introducing any extra (global) strong contextuality.
\begin{proposition}\label{prop:strctx}
  Let $\M \preceq \M'$, let $e$ be a model on $\M$, and suppose that $e$ is extendable to $\M'$.
  Then $e'$ is strongly contextual if and only if $e$ is strongly contextual.
\end{proposition}
\begin{proof}
It is enough to show that $S_e(X) = S_{e'}(X)$; i.e.~that the sets of global assignments consistent with each model coincide.
For a global assignment $s \in \E(X)$,
\begin{calculation}
s \in S_e(X)
\ejust\Leftrightarrow
\Forall{U \in \down\M} \, s|_U \in S_{e}(U)
\just\Leftrightarrow{
``$\Leftarrow$'': all $U$s above satisfy $U \in \down\M \subseteq \down \M'$, hence are covered by some $C \in \M'$}
\Forall{C \in \M'} \, \Forall{W \subseteq C, W \in \down\M} \, s|_W \in S_{e}(W)
\ejust\Leftrightarrow
\Forall{C \in \M'} \, s_C \in S_{e'}(C)
\ejust\Leftrightarrow
s \in S_{e'}(X).
\end{calculation}
\end{proof}

The situation here is more complicated than in the usual case of extensions to the top cover. The key issue is whether compatibility (no-signalling) holds for the extended model. This is by no means guaranteed.
It might happen that the canonical construction $e'$ has too many possible assignments, causing it to be signalling. The reason is that $e'$ picks out all the assignments that are `locally' consistent with $e$, but when overlaps of contexts arise that were not contained in the original cover it is possible that assignments are not compatible. We give an example to illustrate how such behaviour might arise.

\begin{example}\label{ex:sig}
The model $e$ on the cover $\M = \{AB, BC,CD,DA\}$ is given by the following table.
\begin{center}
\begin{tabular}{p{3pt}l|cccc}
&~& $00$ & $01$ & $10$ & $11$   \\ \hline
$A$&$B$ & $1$ & $1$ & $1$ & $1$ \\
$B$&$C$ & $1$ & $0$ & $0$ & $1$ \\
$C$&$D$ & $1$ & $0$ & $0$ & $1$ \\
$D$&$A$ & $1$ & $1$ & $1$ & $1$
\end{tabular}
\end{center}
We consider the canonical extension $e'$ to the cover $\M' = \{ABC,BCD\} \succeq \M$ (see Figure \ref{fig:sig}).
\begin{center}
\begin{tabular}{p{3pt}p{3pt}l|cccccccc}
&~&~      & $000$ & $001$ & $010$ & $011$ & $100$ & $101$ & $110$ & $111$   \\ \hline
$A$&$B$&$D$ &  $1$  &  $1$  &  $1$  &  $1$  &  $1$  &  $1$  &  $1$  &  $1$  \\
$B$&$C$&$D$ &  $1$  &  $0$  &  $0$  &  $0$  &  $0$  &  $0$  &  $0$  &  $1$  
\end{tabular}
\end{center}
The extension $e'$ is clearly not compatible. For example,
\[
1 = \left.{e'_{ABD}}\right|_{BD}(01) \neq \left.{e'_{BCD}}\right|_{BD}(01) = 0.
\]

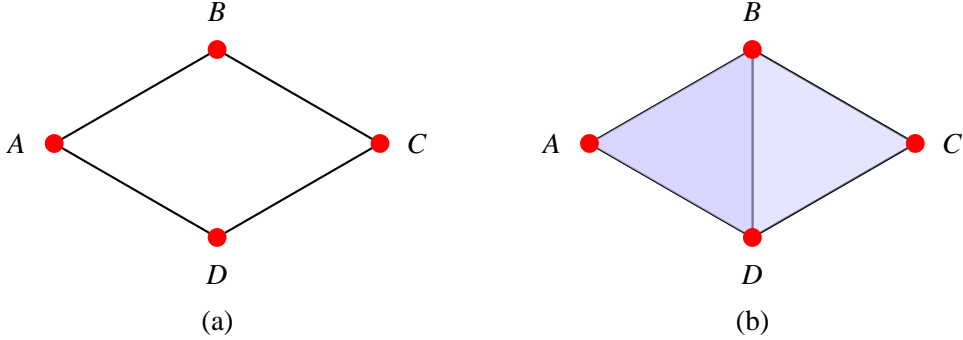
\begin{figure}
\caption{\label{fig:sig}\index{measurement scenario!as simplicial complex} The measurement covers of Example \ref{ex:sig}: (a) $\M = \{AB,BC,CD,DA\}$ (b) $\M' = \{ABD,BCD\}$.}
\begin{center}
\begin{tabular}{cc}
\begin{tikzpicture}[scale=2.5]
  \path (-0.866, 0.5 ) coordinate (A);
  \path ( 0, 1 ) coordinate (B);
  \path ( 0.866, 0.5 ) coordinate (C);
  \path ( 0, 0 ) coordinate (D);

  \begin{scope}[ thick,black]
    \draw (A) -- (B) -- (C) -- (D) -- cycle;
  \end{scope}
 
  \fill[red]   (A) circle(.05) ;
  \fill[red]   (B) circle(.05) ;
  \fill[red] (C) circle(.05) ;
  \fill[red] (D) circle(.05) ;
  
  \path (A) +(-.2 ,0  ) node {$A$}; 
  \path (B) +(0 ,.2  ) node {$B$};
  \path (C) +(.2 ,0) node {$C$};
  \path (D) +(0 ,-.2) node {$D$};
\end{tikzpicture}
\qquad & \qquad
\begin{tikzpicture}[scale=2.5]
  \path (-0.866, 0.5 ) coordinate (A);
  \path ( 0, 1 ) coordinate (B);
  \path ( 0.866, 0.5 ) coordinate (C);
  \path ( 0, 0 ) coordinate (D);

  \begin{scope}[ thick,black]
    \draw (A) -- (B) -- (D) -- cycle;
    \draw (B) -- (C) -- (D) -- cycle;
  \end{scope}
  
  \begin{scope}[opacity=0.5, color=blue!20]
  \fill (A) -- (B) -- (D) -- cycle;
  \fill (B) -- (C) -- (D) -- cycle;
  \end{scope}
 
  \fill[red]   (A) circle(.05) ;
  \fill[red]   (B) circle(.05) ;
  \fill[red] (C) circle(.05) ;
  \fill[red] (D) circle(.05) ;
  
  \path (A) +(-.2 ,0  ) node {$A$}; 
  \path (B) +(0 ,.2  ) node {$B$};
  \path (C) +(.2 ,0) node {$C$};
  \path (D) +(0 ,-.2) node {$D$};
   
\end{tikzpicture}
\\ (a) \qquad & \qquad (b)
\end{tabular}
\end{center}
\end{figure}

\end{example}

\subsection*{Sub-models}

No-signalling extensions can also be related to the contextuality of sub-models of an empirical model.

\begin{definition}\index{empirical model!sub-model}
Let $e$ be a model on $\M$. For any $U \subseteq X$, the \emph{induced sub-model of $e$ on $U$} is $\{ e_{U \cap C} \}_{C \in \M}$.
\end{definition}
By compatibility of the original model it is clear that any induced sub-model will be a well-defined empirical model. We note that this definition holds for any empirical model, not just the possibilistic ones we are concerned with in this section.

\begin{proposition}\index{contextuality!of sub-models}
Let $e$ be a model on $\M$. If $f$ extends $e$ to $\M'$ then all of the sub-models induced by $\M'$ are non-contextual.
\end{proposition}

\begin{proof}
We show that if $\{f_{C}\}_{C\in\M'}$ extends $e$ to $\M'$ then each $f_C$ is a global section of the induced sub-model of $e$ on $C$. The induced sub-model on $C$ is a model defined on the measurement cover $(C,O,\M|_{C})$ where $\M|_{C} = \{D\cap C\}_{D \in \M}$. Since $f$ is an extension, $f_{C}|_{D \cap C} = e_{D \cap C}$ for all $D \in \M$, and $f_{C}$ must be a global section.
\end{proof}

The converse is not necessarily true, however. It is possible that all sub-models that are induced in this way by elements of a cover $\M' \succeq \M$ have a global section
but that one cannot find an extension $\{f_C\}_{C\in\M'}$ (canonical or otherwise) that is no-signalling. This is the situation for example \ref{ex:sig}, for which the induced sub-models on $ABD$ and $BCD$ are non-contextual, but, as we have seen, the model cannot be extended to the cover $\M' = \{ABD,BCD\}$.

Nevertheless, more can be said about the relationship between extendability and induced sub-models when we talk of the strong properties.
\begin{corollary}
  Let $\M \preceq \M'$ and let $e$ be a model on $\M$. Then $e$ is strongly  non-extendable
 to $\M'$ if and only if there exists some $C \in \M'$ such that the induced sub-model of $e$ on $C$
 is strongly contextual.
  In particular, $e$ is strongly non-extendable to $\M' = \mathcal{P}_{n(\M)}X$ if and only if it has a strongly contextual induced sub-model of size $n(\M)$.
\end{corollary}

\begin{proof}
This follows from Proposition \ref{prop:smaller}.
\end{proof}

\section{Construction of Bell-type Models}\label{ssec:construction}

We turn now to a second construction, which builds
from empirical models on certain kinds of measurement covers Bell models that are equivalent in terms of contextuality. Note that this construction is not restricted to possibilistic models, and will work for any kind of empirical model. The idea is to start with a model on some measurement scenario $(X,O,\M)$, and to transform this into a model on the Bell scenario $\left(\coprod_{i=1}^n X,O,\prod_{i=1}^n  X\right)$. Measurements here are of the form $\tuple{x,i}$ with $x \in X$ and $i \in \{1,\ldots,n\}$ identifying the site (copy of $X$).

\begin{definition}\index{assignment!codiagonal}
  Let $U \subseteq \coprod_{i=1}^n X$. An assignment $s \in \E(U)$
is said to be \emph{codiagonal} if it satisfies
\[\Forall{x \in X} \, \Forall{i,j \in \{1,\ldots, n\}} \, \tuple{x,i}, \tuple{x,j} \in U \implies s(\tuple{x,i}) = s(\tuple{x,j});\]
i.e.~copies of the same measurement at different sites are assigned the same outcome.
Equivalently, in categorical terms, an assignment $s: U \rightarrow O$ is codiagonal when it factors as
\[\xymatrix{U \ar@{^{(}->}[r] & \coprod_{i=1}^n \ar[r]^{~\;\;\bigtriangledown_{_n}} X & X \ar[r] & O}.\]
We denote the set of such assignments by $\Ecd(U)$.
\end{definition}

With each set $U \subseteq \coprod_{i=1}^n X$ of measurements on the new scenario,
we associate a set $\underline{U} \subseteq X$ of measurements on the original, which is obtained by forgetting the site information; i.e.
$\underline{U} := \setdef{x \in X}{\Exists{i} \tuple{x,i} \in U}$.
It is clear that there is a bijection $\Ecd(U) \cong \E(\underline{U})$, which commutes with restrictions to smaller contexts, and we write $\underline{s}$ for the image of an assignment $s \in \Ecd(U)$ under this map. Recall also that $\mathcal{P}_nX := \setdef{Y \subseteq X}{|Y| = n}$.

\begin{definition}
With any empirical model $f$ defined on a measurement scenario $(X,O,\mathcal{P}_nX)$ we associate an $n$-partite empirical model $f^\Bell$ on the Bell scenario $(\coprod_{i=1}^n X,O,\prod_{i=1}^n  X)$ defined by
  \[f^\Bell_C(s) =
                \begin{cases}
		  f_{\underline{C}}(\underline{s})
		  & \text{if $s \in \Ecd(C)$}
		  \\
		  0
		  & \text{if $s \not\in \Ecd(C)$}
		\end{cases} .
  \]
\end{definition}

\begin{proposition}
  $f^\Bell$ is an empirical model.
\end{proposition}
\begin{proof}
In general, we see that $f^\Bell_C$ is equivalent to $f_{\underline{C}}$. Since $\underline{C}$ is a context (not necessarily maximal) of the measurement scenario for $f$, then $f^\Bell_C$ is a well-defined distribution.
As for compatibility, let $C_1$ and $C_2$ be two maximal contexts, and let $C_{1,2} := C_1 \cap C_2$. Then, for $t \in \E(C_{1,2})$,
{\allowdisplaybreaks
\begin{calculation}
  \left. f^\Bell_{C_1}\right|_{C_{1,2}}(t)
  \just={definition of marginalisation}
\sum_{s \in \E(C_1) \atop s|_{C_{1,2}} = t} \, f^\Bell_{C_1}(s)
\just={since $f^\Bell_{C_1}(s)\neq 0$ only if $s$ is codiagonal}
\sum_{s \in \Ecd(C_1) \atop s|_{C_{1,2}} = t} \, f_{\underline{C_1}}(\underline{s})
\just={using the bijection $\Ecd(C_1) \cong \E(\underline{C_1}) :: s \mapsto \underline{s}$}
\sum_{\underline{s} \in \E(\underline{C_1}) \atop \underline{s}|_{\underline{C_{1,2}}} = \underline{t}} \, f_{\underline{C_1}}(\underline{s})
\just={compatibility condition for the model $f$}
\sum_{\underline{r} \in \E(\underline{C_2}) \atop \underline{r}|_{\underline{C_{1,2}}} = \underline{t}} \, f_{\underline{C_2}}(\underline{r}) \\
\just={same steps in reverse order for $C_2$}
\left. f^\Bell_{C_2}\right|_{C_{1,2}}(t).
\end{calculation}} \\
We conclude that
\[ \left. f^\Bell_{C_1}\right|_{C_{1,2}} = \left. f^\Bell_{C_2}\right|_{C_{1,2}}.\]
\end{proof}

\begin{proposition}
  There is a bijection between the global sections of $f$ and of $f^\Bell$.
  In particular,
  \begin{itemize}
  \item
  $f^\Bell$ is non-local if and only if $f$ is contextual,
  \item
  $f^\Bell$ is logically non-local if and only if $f$ is logically contextual,
  \item
  $f^\Bell$ is strongly non-local if and only if $f$ is strongly contextual.
  \end{itemize}
\end{proposition}
\begin{proof}
By the definition of $f^\Bell$, it is clear that any assignments that are deemed possible by a global section $d \in \mathcal{D}_R\E(\coprod_{i=1}^nX)$ of the model must belong to $\Ecd(\coprod_{i=1}^nX)$. Recall that there exists a bijection $\Ecd(U)\cong\E(\underline{U}), s \mapsto \underline{s}$
  for each $U\subseteq \coprod_{i=1}^nX$, and that these bijections commute with restrictions.
  The correspondence lifts to provide bijections between $\mathcal{D}_R\Ecd(U)$ and $\mathcal{D}_R\E(\underline{U})$
that commute with marginalisation. Therefore, if $U = \coprod_{i=1}^nX$ (and $\underline{U} = X$), the resulting bijection gives the desired correspondence between global sections.
\end{proof}

On a related note, it is worth pointing out that any compatible family of $\mathbb{Z}$-linear combinations of assignments in $f$ can be lifted to a compatible family of this kind on $f^\Bell$ via the injective maps
\[
\E(\underline{C}) \cong \Ecd(C) \subseteq \E(C).
\]
This straightforwardly leads to the following proposition, which relates to the cohomological witness of non-locality/contextuality from \cite{abramsky:11a}. It is unclear whether the converse holds.

\begin{proposition}
The existence of a non-vanishing cohomological obstruction in $f^\Bell$ implies the existence of a non-vanishing obstruction in $f$\index{cohomology!of constructed models}.
\end{proposition}

In good cases, we can use the canonical extension of the previous section to extend a model $e$ on $(X,O,\M)$ to a model $e'$ on $(X,O,\mathcal{P}_nX)$ and then apply the present construction to obtain a Bell model which, in general, is equivalent to the original in terms of strong contextuality.

\begin{corollary}\label{cor:sccor}
  Let $e$ be an empirical model on $(X,O,\M)$ and suppose that $e$ is canonically extendable to $(X,O,\mathcal{P}_nX)$.
  Then $e'^\Bell$ is strongly non-local if and only if $e$ is  strongly contextual.
\end{corollary}

\section{Bell Models from Kochen-Specker Models}\label{ssec:ksconstruction}

The construction of an equivalent Bell model can be carried out for all Kochen-Specker models in which maximal contexts are all of the same size.

\begin{proposition}\label{prop:ksce}
The Kochen-Specker model for any scenario\footnote{Recall that we must have $|C| = n$ for all $C \in \M$.} $(X,O,\M)$ is canonically extendable to $(X,O,\mathcal{P}_nX)$.
\end{proposition}
\begin{proof}
  Let $e$ be the Kochen-Specker model for the scenario $(X,O,\M)$.
  For a context $C \in \mathcal{P}_nX$ and an assignment $s \in \E(C)$, we have
  \[ e'_C(s) = 1 \;\;\leftrightarrow\;\; s \in S_{e}(C) \;\;\leftrightarrow\;\;
     \Forall{W \subseteq C, W \in \down\M} \, o(s|_W) \leq 1 \;\wedge\;
     \Forall{D \subseteq C, D \in \M}  \,    o(s|_D) = 1 
  \]
  Since $C \in \mathcal{P}_nX$, it cannot have any proper sub-context that is a maximal context in $\M$.
  Thus, we can write:
  \[ e'_C(s) = \begin{cases}
                 1  & \text{if $C \in \M \wedge o(s) = 1$}
		 \\
		 1  & \text{if $C \notin \M \wedge \Forall{W \subseteq C, W \in \down\M} \, o(s|_W) \leq 1$}
		 \\
		 0  & \text{otherwise}
	       \end{cases}.
  \]
  This is clearly a distribution for every $C \in \mathcal{P}_nX$, as there is always at least one possible assignment.

  Now, let $U \subsetneq C$ and consider the marginalisation $\left. e'_C\right|_U$.
 First, we look at the case that $C \notin \M$. Then it is easy to see that, for any $t \in \E(U)$, $\left. e'_C\right|_U(t) = 1$ implies $\Forall{V \subseteq U, V \in \down\M} o(t|_V) \leq 1$.
  Conversely, if the latter holds, one can extend $t$ to $C$
  by assigning $0$ to all new measurements, obtaining an assignment $s \in \E(C)$ that satisfies
  $\Forall{W \subseteq C, W \in \down\M} o(s|_W) \leq 1$, since $o(s|_W) = o(s|_{W \cap U}) = o(t|_{W\cap U}) \leq 1$ for all such $W$. We then have $e'_C(s)=1$, and hence $\left. e'_C\right|_U(t) = \left.e'_C\right|_U(s|_U) =1$. So,
  \begin{equation}\label{eq:aaa}
    \left.e'_C\right|_U(t)=1 \quad \leftrightarrow\quad \Forall{V \subseteq U, V \in \down\M} \, o(t|_V)\leq 1.
  \end{equation}
  For the case that $C \in \M$, a section $t \in \E(U)$ satisfies $\left. e'_C\right|_U(t) = 1$ if and only if $o(t) \leq 1$.
  So (\ref{eq:aaa}) holds in this case, too, since $U$ itself is one of the $V$'s in the formula.

  This shows that the marginalisation to any $U$ is independent of the maximal context from which one starts, proving compatibility as required.
\end{proof}

Combining this with corollary \ref{cor:sccor} gives the following.

\begin{corollary}\label{cor:kscor2}\index{Kochen-Specker model!equivalent Bell model}
For any Kochen-Specker model $e$, the Bell model $e'^\Bell$ is well-defined, and is strongly non-local if and only if $e$ is strongly contextual.
\end{corollary}

It was shown in \cite{mansfield:13t} that symmetric Kochen-Specker models are contextual if and only if they are strongly contextual\footnote{A Kochen-Specker model is symmetric if for every pair of measurements $m,m' \in X$ there exists a hypergraph automorphism $\alpha:X \rightarrow X$ of the hypergraph $(X,\M)$ such that $\alpha(m) = m'$.}. This means that for the whole class of symmetric Kochen-Specker models (including the contextual triangle and the $18$-vector model), we can construct Bell models that are equivalent in terms of contextuality.

\begin{corollary}
For any symmetric Kochen-Specker model $e$, the Bell model $e'^\Bell$ is well-defined, and is non-local if and only if $e$ is contextual.
\end{corollary}

\section{Examples}\label{examples}

\subsection*{The Contextual Triangle}\index{contextual triangle}
Carrying out the construction on the triangle yields the following $(2,3,2)$ model.
\begin{center}
\begin{tabular}{p{3pt}l|cccc}
&~& $00$ & $01$ & $10$ & $11$   \\ \hline
$A_{1}$&$A_{2}$ & 1 & 0 & 0 & 1 \\
$A_{1}$&$B_{2}$ & 0 & 1 & 1 & 0 \\
$A_{1}$&$C_{2}$ & 0 & 1 & 1 & 0 \\
$B_{1}$&$A_{2}$ & 0 & 1 & 1 & 0 \\
$B_{1}$&$B_{2}$ & 1 & 0 & 0 & 1 \\
$B_{1}$&$C_{2}$ & 0 & 1 & 1 & 0 \\
$C_{1}$&$A_{2}$ & 0 & 1 & 1 & 0 \\
$C_{1}$&$B_{2}$ & 0 & 1 & 1 & 0 \\
$C_{1}$&$C_{2}$ & 1 & 0 & 0 & 1
\end{tabular}
\end{center}
We include site indices to make it clear that different measurements in the same context are now considered to belong to different sites.
It is especially interesting that the model can be seen to contain many different PR boxes \cite{popescu:94} as sub-models. These are:
\begin{center}
\begin{tabular}{ccc}
\begin{tabular}{p{3pt}l|cccc}
&~& $00$ & $01$ & $10$ & $11$   \\ \hline
$A_1$&$A_2$ & $1$ & $0$ & $0$ & $1$ \\
$A_1$&$C_2$ & $0$ & $1$ & $1$ & $0$ \\
$B_1$&$A_2$ & $0$ & $1$ & $1$ & $0$ \\
$B_1$&$C_2$ & $0$ & $1$ & $1$ & $0$
\end{tabular}
&
\begin{tabular}{p{3pt}l|cccc}
&~& $00$ & $01$ & $10$ & $11$  \\ \hline
$B_1$&$A_2$ & $0$ & $1$ & $1$ & $0$ \\
$B_1$&$B_2$ & $1$ & $0$ & $0$ & $1$ \\
$C_1$&$A_2$ & $0$ & $1$ & $1$ & $0$ \\
$C_1$&$B_2$ & $0$ & $1$ & $1$ & $0$ 
\end{tabular}
&
\begin{tabular}{p{3pt}l|cccc}
&~& $00$ & $01$ & $10$ & $11$   \\ \hline
$B_1$&$A_2$ & $0$ & $1$ & $1$ & $0$ \\
$B_1$&$C_2$ & $0$ & $1$ & $1$ & $0$ \\
$C_1$&$A_2$ & $0$ & $1$ & $1$ & $0$ \\
$C_1$&$C_2$ & $1$ & $0$ & $0$ & $1$
\end{tabular}
\end{tabular}
\end{center}
and those obtained by reversing the order of the measurements.

We note that neither the triangle nor the PR box is realisable in quantum mechanics. The triangle provides the simplest possible example of a contextual model, and the PR box is the only strongly contextual $(2,2,2)$ model \cite{lal:13}.

\subsection*{The Peres-Mermin Square}\index{Peres-Mermin model}

The Peres-Mermin square \cite{mermin:93} is another example of a strongly contextual model, though it does not fall into the class of Kochen-Specker models. It, too, has the desirable property that it can be extended to $\mathcal{P}_nX$. We can therefore construct a tripartite Bell model from it which is equivalent in terms of strong contextuality. The constructed model contains $36$ different $(3,5,2)$ non-local sub-models, which are essentially `padded-out' versions of the square. The following table represents the non-local sub-model on the measurement cover \[\M = \{A,B,C,D,G\} \times \{B,D,E,F,H\} \times \{C,F,G,H,I\},\] though we only explicitly write those rows that do not have full support.
\begin{center}
\begin{tabular}{p{3pt}p{3pt}l|cccccccc}
&~&~      & $000$ & $001$ & $010$ & $011$ & $100$ & $101$ & $110$ & $111$   \\ \hline
$A_1$&$B_2$&$C_3$ &  $0$  &  $1$  &  $1$  &  $0$  &  $1$  &  $0$  &  $0$  &  $1$  \\
$D_1$&$E_2$&$F_3$ &  $0$  &  $1$  &  $1$  &  $0$  &  $1$  &  $0$  &  $0$  &  $1$  \\
$G_1$&$H_2$&$I_3$ &  $0$  &  $1$  &  $1$  &  $0$  &  $1$  &  $0$  &  $0$  &  $1$  \\
$A_1$&$D_2$&$G_3$ &  $1$  &  $0$  &  $0$  &  $1$  &  $0$  &  $1$  &  $1$  &  $0$  \\
$B_1$&$E_2$&$H_3$ &  $1$  &  $0$  &  $0$  &  $1$  &  $0$  &  $1$  &  $1$  &  $0$  \\
$C_1$&$F_2$&$I_3$ &  $1$  &  $0$  &  $0$  &  $1$  &  $0$  &  $1$  &  $1$  &  $0$ 
\end{tabular}
\end{center}
Ignoring the additional rows, and the indices, which are there as a reminder that the measurements belong to different sites, this looks just like the table for the Peres-Mermin model itself; though it is a genuinely new strongly non-local Bell model.

An interesting point is that since the Peres-Mermin contextuality proof is based on a parity argument, just like the GHZ proof, one might expect that the Bell model we have constructed should contain a GHZ sub-model. However, it can easily be shown by comparison with the table for the GHZ-Mermin model that this is not the case. Similarly, it can be shown (for example by comparing models in a three-dimensional tabular representation \cite{mansfield:13t}) that it contains no tripartite Hardy paradox \cite{ghosh:98,wang:12}.

\section{Conclusion}\label{conclusions}

We have dealt with two related ideas. The refinement of extendability introduced here is a development of the sheaf-theoretic framework, which captures partial approximations to locality and non-contextuality. This allows us to characterise contextuality and strong contextuality in sub-models of an empirical model, as we have seen in section \ref{ssec:extension}.

The second idea is to introduce a method of constructing Bell models from models of a more general kind in such a way that these are equivalent in terms of non-locality/contextuality. This can even work at the level of probabilities. Equivalent Bell forms of models are desirable since, both practically and theoretically, it is easier to motivate a notion of locality in such scenarios than the equivalent notion of non-contextuality in a more general scenario, as one can always appeal to relativity as a justification for certain assumptions. We have also mentioned that non-locality is better understood as an information theoretic resource.

These two ideas are related by the fact that, for any model, the existence of the canonical extension to $\mathcal{P}_nX$ will guarantee the ability to construct a Bell model that is equivalent in terms of strong contextuality. We have proved that for the entire class of Kochen-Specker models with maximal contexts of constant size we can carry out this construction, and that for the symmetric models the equivalence even holds for contextuality. Even in the more general form, which applies to strong contextuality only, this is a very useful result since so many of the familiar examples of non-local/contextual models are strongly contextual: we have mentioned the GHZ-Mermin model, the 18-vector Kochen-Specker model, the Peres-Mermin model, and the Popescu-Rohrlich correlations.

There are several open questions arising from this work. We would like to know whether there is an analogue of \Vorobev's theorem \cite{vorobev:62}\index{Vorob'ev's theorem} for this partial notion of extendibility; that is, given any measurement cover $\M'$, can there be a complete characterisation of the measurement covers $\M \preceq \M'$ such that any empirical model defined on $\M$ is extendable to $\M'$. This could potentially lead to applications to macroscopic realism\index{realism!macroscopic realism} similar to \cite{barbosa:13}. We would also like to know whether there are other general classes of `good' models for which we can guarantee the ability to extend to $\mathcal{P}_nX$ and thereby construct Bell models that are equivalent either in terms of contextuality or strong contextuality: a class of Peres-Mermin-like models for example.

It is especially interesting that when we constructed the equivalent Bell model for the contextual triangle, we ended up with what is essentially a folding of PR boxes. This appears to point to a deeper relationship between the models, which merits further investigation. The PR box has been much studied and has been considered, for example, as a possible unit of non-locality \cite{barrett:05}. Since the triangle is the simplest possible example of a contextual model, it is conceivable that it could be a unit of contextuality. One might hope for some sort of analogous result to Kuratowski's theorem\index{Kuratowski's theorem} for graphs, which states that a graph is planar if and only if it does not contain $K_5$ or $K_{3,3}$ as subgraphs. For example, it could be that, for some notion of reducibility, contextual models must reduce to the triangle or to elements of some set of irreducible models containing the triangle.

Another important issue that has not been dealt with so far is that of quantum realisability: given that a model is quantum realisable, we would like to understand when its extensions are and vice versa. Here, it is especially relevant to the example of the Peres-Mermin model. An aim of this work is to propose Bell tests that correspond to contextuality proofs such as that of Peres \& Mermin by giving a means of quantum mechanically reproducing its equivalent Bell model. First, however, it will be necessary to understand how quantum realisability relates to our constructions.

\section*{Acknowledgements}
The authors thank Samson Abramsky for valuable discussions and comments.
RSB gratefully acknowledges support from the Marie Curie Initial Training Network - MALOA - From MAthematical LOgic to Applications, PITN-GA-2009-238381.

\bibliographystyle{eptcs}
\bibliography{refs}

\end{document}